\documentclass[11pt]{article}
\usepackage[margin=1.05in]{geometry}
\usepackage{amsmath,amssymb,amsthm,bm}
\usepackage{graphicx}
\usepackage{booktabs}
\usepackage[round]{natbib}
\usepackage[colorlinks=true,linkcolor=blue,citecolor=blue,urlcolor=blue]{hyperref}

\newtheorem{proposition}{Proposition}
\newtheorem{corollary}{Corollary}
\newtheorem{definition}{Definition}
\newtheorem{remark}{Remark}

\newcommand{\E}{\mathbb{E}}
\newcommand{\R}{\mathbb{R}}
\newcommand{\Var}{\mathrm{Var}}
\newcommand{\Cov}{\mathrm{Cov}}
\newcommand{\N}{\mathcal{N}}
\newcommand{\da}{\bar\alpha}

\title{Denoising Subordinated Probabilistic Models:\\
Diffusion with a Tempered-Stable Volatility Clock,\\
and What the Noise Mechanism Actually Controls}
\author{Junchi Shen, Helin Zhao \thanks{Working draft. All derivations and every number
in this paper are reproduced by the accompanying scripts
\texttt{01\_sampler\_check.py}--\texttt{07\_blind.py} (fixed seeds, single
laptop).}}
\date{July 2026}

\begin{document}
\maketitle

\begin{abstract}
Heavy-tailed extensions of denoising diffusion models replace the Gaussian
noise by a Gaussian \emph{variance mixture} $\varepsilon_i=\sqrt{A_i}z_i$:
denoising L\'evy probabilistic models (DLPM) take the mixing variables
$A_i$ i.i.d.\ across coordinates, Student-$t$ EDM takes a single $A$ shared
by all coordinates. Neither extreme has \emph{dynamics}, yet temporal
dependence of the noise amplitude---volatility clustering---is the defining
stylized fact of the financial data these models target. We introduce the
\emph{Denoising Subordinated Probabilistic Model} (DSPM), whose mixing
vector is a stationary AR(1) chain driven by tempered-stable subordinator
increments---the discrete-time Barndorff-Nielsen--Shephard stochastic
volatility process---running along the data axis. Conditionally on the
chain, the entire DDPM machinery survives verbatim, and we prove closed-form
identities linking the chain parameters to the excess kurtosis and the
squared-noise autocorrelation, yielding an exactly identified,
\emph{analytically invertible} calibration with a falsifiable feasibility
bound $\rho_{r^2}(1)<K/(3(K+2))$. DDPM, DLPM-type stable noise, and
shared-scale Student-$t$-type noise are boundary cases of one memory
parameter. We then prove a result that delimits the entire design space: if
the denoiser is conditioned on $A$ (as sufficiency demands), the mixing law
is a \emph{nuisance}---in the exact-denoiser limit the generated
distribution is invariant to it, and interventions on $A$ have no effect.
Controlled experiments (identical denoiser, schedule, budget; only the law
of $A$ or the conditioning bit varies) confirm both halves: conditioned
models match the data's clustering regardless of the mixing law, and a
designed $\times8$ volatility shock moves the generated amplitude envelope
by less than $13\%$ against a naive prediction of $+183\%$; \emph{blind}
models transmit the noise structure exactly as calibrated---lag-one
squared-return autocorrelation $0.151$ against the calibrated $0.161$ with
the predicted geometric decay, while the i.i.d.\ variant's clustering
collapses to $0.008$. The mixing law controls what structure the noise
\emph{offers}; the conditioning bit controls whether it \emph{transmits}.
Finally we close the loop (v2): coupling the chain to the data through a
variational volatility encoder (with a log-normal volatility prior, for
closed-form KL) restores control---provided the objective retains the
stochastic-volatility likelihood whose log-determinant term the
simplified denoising loss discards; naive $\beta$-weighted couplings
collapse or leak, and both failure modes are documented. With it, the
same $\times8$ intervention scales the envelope by $3.07$ (naive scaling:
$2.83$), the encoder tracks the true latent volatility of held-out paths
(per-path correlation $0.57$; $0.76$ pooled), and the learned prior
memory moves from its moment-calibrated initialization toward the true
persistence. The cost is substantial and stated plainly: the coupled
model's unconditional kurtosis falls below even the Gaussian
baseline---fidelity and control are demonstrated separately, not yet
jointly.
\end{abstract}

\section{Introduction}

Score-based and denoising diffusion models
\citep{sohl2015,ho2020ddpm,song2021sde} corrupt data with Gaussian noise
and learn to invert the corruption. For financial time series the Gaussian
choice discards the domain's structure. Fifty years of empirical finance
have fixed a short list of \emph{stylized facts} of asset returns
\citep{cont2001}: (i) heavy-tailed marginals with tail index
$\zeta\approx3$ \citep{gopikrishnan1999}; (ii) negligible linear
autocorrelation; (iii) strong, slowly decaying autocorrelation of squared
returns---\emph{volatility clustering}; (iv) aggregational Gaussianity.

Recent work addresses fact (i) inside the diffusion framework by replacing
the Gaussian. Through the classical variance-mixture representation
$\varepsilon=\sqrt{A}Z$, the published proposals sit at two extremes of one
axis. Denoising L\'evy probabilistic models (DLPM,
\citealp{shariatian2024dlpm}) and the L\'evy-It\^o model
\citep{yoon2023lim} use $\alpha$-stable noise: mixing variables i.i.d.\
across coordinates, \emph{no memory}. Student-$t$ EDM
\citep{pandey2025tedm} uses multivariate Student-$t$ noise: one
inverse-gamma mixing variable per sample, \emph{infinite memory}. Neither
can represent fact (iii) in its noise mechanism, and financial econometrics
identified the missing object decades ago: the mixing variable is the flow
of economic time \citep{clark1973}, it is \emph{persistent}, and its
canonical continuous-time model is an Ornstein--Uhlenbeck process driven by
a (tempered) L\'evy subordinator---Barndorff-Nielsen--Shephard (BNS)
stochastic volatility \citep{bns2001}, embedded in the time-changed L\'evy
consensus of \citet{carr2003sv}, with exponential tempering
\citep{koponen1995,cgmy2002} curing the infinite variance of pure stable
laws \citep{mantegna1994}.

This paper makes two contributions that pull in deliberately opposite
directions, and we believe the tension is the point.

\paragraph{Constructive contribution.} We define the DSPM
(Section~\ref{sec:model}): a DDPM whose per-coordinate noise variance is a
stationary tempered-stable AR(1) chain along the \emph{data} axis.
Conditionally on the chain the Gaussian posterior, the closed-form ELBO and
ancestral sampling hold verbatim (Propositions~\ref{prop:posterior},
\ref{prop:elbo})---the DDPM derivations never use independence of the
mixing variables. The excess kurtosis and squared-noise autocorrelation of
the induced noise are closed-form in the chain parameters
(Proposition~\ref{prop:moments}); the calibration map is exactly identified
with an analytic inverse and a falsifiable feasibility bound
(Corollary~\ref{cor:calib}); and DDPM, DLPM-type and t-EDM-type noise are
boundary cases of the single memory parameter $\varphi$
(Proposition~\ref{prop:limits}). The noise prior coincides with the
``diffusing diffusivity'' models of statistical physics
(Section~\ref{sec:physics}).

\paragraph{Delimiting contribution.} We then prove
(Proposition~\ref{prop:invariance}) that when the denoiser is conditioned
on $A$---the statistically natural choice, since $A$ is known at both
training and sampling time---the mixing law is a \emph{nuisance parameter}:
in the exact-denoiser limit the generated distribution does not depend on
it at all, and interventions on $A$ provably do nothing. The experiments
(Section~\ref{sec:exp}) confirm both the theorem and its boundary with
unusual sharpness: conditioned models reproduce clustering equally well
whatever the law of $A$ (even pure DDPM), a $\times8$ engineered volatility
shock leaves the generated envelope essentially unchanged, while
\emph{blind} models (denoiser not shown $A$) transmit the mechanism's
structure quantitatively---the calibrated lag-one autocorrelation appears
in the samples to within $0.01$, with the geometric decay fingerprint of
Proposition~\ref{prop:moments}, and vanishes for the i.i.d.\ variant. The
practical moral for the growing heavy-tailed-diffusion literature
\citep{heavytails2026}: \emph{the mixing law controls what the noise
offers; the conditioning choice controls whether it transmits}; claimed
benefits of heavy-tailed noise in the conditioned regime must be
inductive-bias effects, not mechanism effects---and in our measurements,
across three training seeds, they are at most marginal even for tail
fidelity and absent for dependence structure.

\paragraph{Coupling contribution (v2).} The invariance theorem names its
own escape route: break the independence of $x_0$ and $A$. We do so by
variational inference---an encoder posits each path's own latent
volatility, $q_\psi(A\mid x_0)$, trained jointly with the denoiser
(Section~\ref{sec:v2}). The exercise surfaces a second delimiting fact:
the \emph{simplified} denoising loss omits the log-determinant of the
noise covariance---exactly the term through which a Gaussian likelihood
fits a variance---and the whitened residual it retains is nearly
scale-free in $A$, leaving the encoder only weak indirect gradient
pathways; empirically these do not suffice (both failure modes are
documented in Remark~\ref{rem:logdet}). Restoring the term (equivalently,
adding the exact variational objective of the classical stochastic-%
volatility state-space model) turns the encoder into a volatility
filter and unlocks what v1 provably lacks: interventions on $A$ now steer
generation quantitatively, and the prior's memory parameter, learned by
empirical Bayes, moves in the direction of the ground-truth
persistence that moment calibration underestimates
(Section~\ref{sec:exp6}).

\section{Background}\label{sec:background}

\paragraph{DDPM.} For $x_0\in\R^d$ and a schedule
$\beta_1,\dots,\beta_T\in(0,1)$, DDPM \citep{ho2020ddpm} sets
$q(x_t\mid x_{t-1})=\N(\sqrt{1-\beta_t}\,x_{t-1},\beta_t I)$, giving
$q(x_t\mid x_0)=\N(\sqrt{\da_t}x_0,(1-\da_t)I)$ with $\alpha_t=1-\beta_t$,
$\da_t=\prod_{s\le t}\alpha_s$; training minimizes the reweighted ELBO in
the $\varepsilon$-parametrization.

\paragraph{Variance mixtures.} If $A>0$ is random and $Z\sim\N(0,1)$
independent, $\varepsilon=\sqrt{A}Z$ has
$\E e^{iu\varepsilon}=\E e^{-u^2A/2}$: one-sided $\alpha$-stable $A$ gives
symmetric $2\alpha$-stable $\varepsilon$ (DLPM); inverse-gamma $A$ gives
Student-$t$ (t-EDM); lognormal $A$ gives Clark's
\citeyearpar{clark1973} model. The mixing variable is the instantaneous
variance; its temporal dependence \emph{is} volatility clustering.

\paragraph{Tempered stable subordinators.} An increment
$\eta\sim\mathrm{TS}(\alpha,\theta,\delta)$, $0<\alpha<1$, $\theta>0$,
$\delta>0$, is the exponential tilting of the one-sided stable law:
\begin{equation}\label{eq:lt}
\E\,e^{-u\eta}
=\exp\!\bigl(-\delta[(\theta+u)^{\alpha}-\theta^{\alpha}]\bigr),
\qquad
\kappa_n(\eta)=\delta\,\alpha\,
\tfrac{\Gamma(n-\alpha)}{\Gamma(1-\alpha)}\,\theta^{\alpha-n};
\end{equation}
in particular $\kappa_1=\delta\alpha\theta^{\alpha-1}$,
$\kappa_2=\delta\alpha(1-\alpha)\theta^{\alpha-2}$. All moments are finite
for $\theta>0$; $\theta\to0$ recovers the stable law
$\E e^{-u\eta}=e^{-\delta u^\alpha}$.

\section{The model}\label{sec:model}

\subsection{Definition}

\begin{definition}[DSPM]\label{def:dspm}
Fix $0<\alpha<1$, $\theta>0$, $0\le\varphi<1$ and a DDPM schedule
$\{\beta_t\}$. The \emph{variance chain} $A=(A_1,\dots,A_d)$ is the
stationary solution of
\begin{equation}\label{eq:chain}
A_i=\varphi A_{i-1}+\eta_i,\qquad
\eta_i\overset{\mathrm{iid}}\sim\mathrm{TS}(\alpha,\theta,\delta),\qquad
\delta=\frac{(1-\varphi)\theta^{1-\alpha}}{\alpha}
\quad(\Rightarrow\ \E A_i=1),
\end{equation}
and, writing $D_A=\mathrm{diag}(A)$, the forward process conditionally on
$A$ is
\begin{equation}\label{eq:forward}
q(x_t\mid x_{t-1},A)=\N\!\bigl(\sqrt{1-\beta_t}\,x_{t-1},\,\beta_tD_A\bigr),
\qquad
q(x_t\mid x_0,A)=\N\!\bigl(\sqrt{\da_t}x_0,\,(1-\da_t)D_A\bigr).
\end{equation}
\end{definition}

The index $i$ is \emph{calendar time of the generated path}, not the
diffusion step: the terminal prior $x_T\approx D_A^{1/2}z$ is a discrete
BNS stochastic-volatility path. Recursion \eqref{eq:chain} is the exact
skeleton of the BNS variance process
$dV_s=-\lambda V_sds+dL_{\lambda s}$ with $\varphi=e^{-\lambda\Delta}$;
fast exact simulation is standard \citep{sabino2022}. Generation is a
latent-variable model: draw $A$ from \eqref{eq:chain}, then run the learned
reverse diffusion given $A$. In training, $x_0$ (data) and $A$ are drawn
\emph{independently}---a fact that Proposition~\ref{prop:invariance} will
make consequential.

\subsection{Conditional Gaussian structure}

\begin{proposition}[Posterior]\label{prop:posterior}
Conditionally on $A$, for $t\ge2$,
\begin{equation}\label{eq:post}
q(x_{t-1}\mid x_t,x_0,A)
=\N\bigl(\tilde\mu_t(x_t,x_0),\ \tilde\beta_tD_A\bigr),\qquad
\tilde\beta_t=\frac{1-\da_{t-1}}{1-\da_t}\beta_t,
\end{equation}
with the \emph{unchanged} DDPM mean
$\tilde\mu_t=\frac{\sqrt{\da_{t-1}}\beta_t}{1-\da_t}x_0
+\frac{\sqrt{\alpha_t}(1-\da_{t-1})}{1-\da_t}x_t$.
\end{proposition}

\begin{proof}
Every covariance in \eqref{eq:forward} is a multiple of the same diagonal
$D_A$, so the computation factorizes over coordinates. In coordinate $i$
with variance factor $a=A_i$: posterior precision
$\frac{\alpha_t}{\beta_ta}+\frac{1}{(1-\da_{t-1})a}
=\frac{1-\da_t}{\beta_t(1-\da_{t-1})a}$, i.e.\ variance
$\tilde\beta_ta$; in the mean, $a$ multiplies prior and likelihood
variances alike and cancels. Only conditioning on the whole vector $A$ is
used---its dependence structure across $i$ is irrelevant.
\end{proof}

\begin{proposition}[ELBO]\label{prop:elbo}
With
$p_\phi(x_{t-1}\mid x_t,A)=\N(\mu_\phi,\tilde\beta_tD_A)$,
$\mu_\phi=\frac{1}{\sqrt{\alpha_t}}\bigl(x_t
-\frac{\beta_t}{\sqrt{1-\da_t}}D_A^{1/2}\hat z_\phi(x_t,t,A)\bigr)$, the
per-step KL of the variational bound is
\begin{equation}\label{eq:kl}
\mathrm{KL}\bigl(q(x_{t-1}\mid x_t,x_0,A)\,\|\,p_\phi\bigr)
=\frac{\beta_t^2}{2\tilde\beta_t\alpha_t(1-\da_t)}
\bigl\|z-\hat z_\phi\bigr\|^2,
\end{equation}
where $z$ is the whitened noise in \eqref{eq:forward}: the
$D_A$-Mahalanobis norm collapses to the Euclidean norm of the whitened
residual, and $\E\|z-\hat z_\phi\|^2$ is the standard reweighted ELBO,
exactly as in DDPM.
\end{proposition}

\begin{proof}
Equal covariances give
$\mathrm{KL}=\|\tilde\mu_t-\mu_\phi\|_{D_A^{-1}}^2/(2\tilde\beta_t)$;
substituting $x_0=(x_t-\sqrt{1-\da_t}D_A^{1/2}z)/\sqrt{\da_t}$ leaves
$\tilde\mu_t-\mu_\phi=\frac{\beta_t}{\sqrt{\alpha_t(1-\da_t)}}
D_A^{1/2}(\hat z_\phi-z)$, and $D_A^{1/2}$ cancels against $D_A^{-1}$.
\end{proof}

Training and sampling are therefore verbatim DDPM with $D_A^{1/2}$
scalings; we condition the denoiser on $A$ and whiten its input,
$\hat z_\phi\bigl(x_t/\sqrt{\da_t+(1-\da_t)A},\,\sqrt A,\,t\bigr)$
(coordinate-wise operations). What does \emph{not} survive is the DLPM
two-draw training acceleration, which needs additive closure of stable
laws \emph{plus} independence; with a correlated chain one simulates the
full $O(d)$ scalar recursion \eqref{eq:chain} per example---cheap in
practice ($\approx10\%$ of wall-clock per training step at our scale:
$111$s vs.\ $100$s per $6{,}000$ steps).

\subsection{Closed-form stylized facts, exact calibration}

Let $r_i=\sqrt{A_i}z_i$, $z\sim\N(0,I)\perp A$: the terminal noise, i.e.\
what the mechanism offers before any learning.

\begin{proposition}[Moment identities]\label{prop:moments}
Under Definition~\ref{def:dspm}, with
$v:=\Var(A_i)=\dfrac{1-\alpha}{\theta(1+\varphi)}$:
(i) $\Cov(A_i,A_{i+h})=v\varphi^h$;
(ii) $\E r_i=0$, $\E r_i^2=1$, $\Cov(r_i,r_{i+h})=0$, and excess kurtosis
\begin{equation}\label{eq:kurt}
K=3v=\frac{3(1-\alpha)}{\theta(1+\varphi)};
\end{equation}
(iii) the squared-noise autocorrelation is
\begin{equation}\label{eq:acf}
\rho_{r^2}(h)=\frac{v\,\varphi^h}{3v+2},\qquad h\ge1 .
\end{equation}
\end{proposition}

\begin{proof}
(i) $A_i=\sum_{k\ge0}\varphi^k\eta_{i-k}$ gives
$\E A=\kappa_1(\eta)/(1-\varphi)=1$ by \eqref{eq:chain},
$\Var A=\kappa_2(\eta)/(1-\varphi^2)=v$ after substituting
\eqref{eq:lt}, and the AR(1) autocovariance $v\varphi^h$.
(ii) $\E r^4=\E A^2\E z^4=3(v+1)$, hence \eqref{eq:kurt}; oddness in $z$
kills linear autocorrelation.
(iii) $\Cov(r_i^2,r_{i+h}^2)=\Cov(A_i,A_{i+h})=v\varphi^h$ and
$\Var(r_i^2)=3(v+1)-1=3v+2$.
\end{proof}

\begin{corollary}[Exact calibration; feasibility]\label{cor:calib}
For targets $K^\ast>0$, $\rho_1^\ast\in(0,1)$ and fixed $\alpha$, the map
$(\varphi,\theta)\mapsto(K,\rho_1)$ is exactly identified with analytic
inverse
\begin{equation}\label{eq:invert}
\varphi=\frac{3\rho_1^\ast(K^\ast+2)}{K^\ast},\qquad
\theta=\frac{3(1-\alpha)}{K^\ast(1+\varphi)},
\end{equation}
feasible iff $\rho_1^\ast<K^\ast/(3(K^\ast+2))$. The bound is a
falsifiable prediction of the model class ($\rho_1<1/3$ always); daily
equity indices ($K\gtrsim8$, $\rho_1\approx0.2$) satisfy it.
\end{corollary}

The activity index $\alpha$ is left free by \eqref{eq:invert}: it shapes
the variance marginal beyond its first two moments (we fix $\alpha=0.7$).
For $\theta>0$ all moments of $r$ are finite---\emph{semi-heavy} tails, the
CGMY class \citep{cgmy2002}; as $\theta\downarrow0$ the noise approaches
symmetric $2\alpha$-stable, so small calibrated $\theta$ yields apparent
power-law behaviour over the observable range before tempering takes over,
matching the empirical situation \citep{mantegna1994,gopikrishnan1999}.

\begin{proposition}[Boundary cases]\label{prop:limits}
(i) $\theta\to\infty$: $v\to0$, $A\to1$ in $L^2$---DDPM.
(ii) $\varphi=0$, $\theta\to0$ at fixed $\delta$: $r_i$ i.i.d.\ symmetric
$2\alpha$-stable---the DLPM/LIM noise family.
(iii) $\varphi\to1$ with increments rescaled to hold $(\E A,v)$ fixed:
$\Cov(A_i,A_j)\to v$ for all $i,j$---one frozen scale per path,
$\rho_{r^2}(h)\to v/(3v+2)$ lag-independent: the correlation structure of
t-EDM (whose inverse-gamma marginal yields Student-$t$).
\end{proposition}

\begin{proof}
(i) is \eqref{eq:kurt}; (ii) is \eqref{eq:lt} with
$\E e^{iur}=\E e^{-u^2A/2}=e^{-\delta(u^2/2)^\alpha}$; (iii) is
\eqref{eq:acf} with $\varphi^h\to1$.
\end{proof}

A single parameter $\varphi\in[0,1]$ thus interpolates between the two
published extremes, and \eqref{eq:invert} pins it to a measurable
quantity---the decay rate of volatility clustering.

\subsection{The nuisance-invariance theorem}\label{sec:invariance}

The results above describe what the noise \emph{offers}. The next result
determines when the generated distribution can \emph{accept} the offer.
Recall that in training $x_0\sim q$ and $A$ are independent.

\begin{proposition}[Nuisance invariance]\label{prop:invariance}
Suppose the reverse chain uses the exact conditional kernels: it is
initialized at $q(x_T\mid A)$ and transitions by
$q(x_{t-1}\mid x_t,A)=\int q(x_{t-1}\mid x_t,x_0,A)\,
q(dx_0\mid x_t,A)$. Then for almost every $A$ the chain terminates at
\begin{equation}
p(x_0\mid A)\;=\;q(x_0\mid A)\;=\;q(x_0).
\end{equation}
Consequently, in the exact-denoiser limit: (i) interventions on $A$
(scenario designs, distributional changes at sampling time within the
support) leave the generated law unchanged; (ii) the law of $A$ does not
affect the generated distribution at all. The mixing law acts only through
what the limit does not cover: finite-$T$ prior mismatch, the
approximation and estimation error of the learned denoiser, and the
optimization path.
\end{proposition}

\begin{proof}
Conditionally on $A$, $(x_0,\dots,x_T)$ is a Markov chain with a
well-defined time reversal; composing the exact reverse kernels from the
exact terminal marginal reproduces the joint law of the forward chain given
$A$, hence the marginal $q(x_0\mid A)$. Independence of $x_0$ and $A$
gives $q(x_0\mid A)=q(x_0)$. (i) and (ii) follow because the terminal law
does not depend on the conditioning value or on the law of $A$.
\end{proof}

\begin{remark}
The proposition formalizes, and sharpens to an exact statement, the recent
skepticism about heavy-tailed noise in diffusion models
\citep{heavytails2026}: in the conditioned regime there is \emph{no
mechanism route} by which the mixing law can shape the learned
distribution; any measured benefit is an inductive-bias effect
(preconditioning of the loss, tail coverage of the terminal prior at
finite $T$). Conversely, a model whose denoiser is \emph{not} shown $A$
cannot whiten it away, and the noise structure transmits. Both predictions
are tested in Section~\ref{sec:exp}.
\end{remark}

\subsection{Coupling the chain to the data (v2)}\label{sec:v2}

Proposition~\ref{prop:invariance} rests on $x_0\perp A$. The coupled model
replaces the independent draw of $A$ by amortized inference. An encoder
maps $x_0$ to a Gaussian mean-field posterior over a \emph{coarse}
log-volatility chain $\mu^c\in\R^{d_c}$ ($d_c=16$ for $d=64$), which is
linearly interpolated to full resolution and exponentiated,
$A=\exp(\mathrm{up}(\mu^c))$. The coarse bottleneck is not a convenience:
it structurally prevents the latent from copying per-coordinate detail of
$x_0$ (sign and fine amplitude), reserving it for what it is meant to
carry---a smooth volatility profile. The prior on $\mu^c$ is the
stationary Gaussian AR(1) with memory $a$ and stationary variance $s^2$
(log-normal OU volatility, i.e.\ \citealp{clark1973}), mean fixed to
$-s^2/2$ so $\E A=1$; $(a,s^2)$ are learned (empirical Bayes), initialized
at the closed-form calibration \eqref{eq:invert} translated to the
log-normal family: under the log-normal chain,
$\Var(A)=e^{s^2}-1=K/3$ and
$\rho_{r^2}(1)=(e^{s^2a_{\mathrm{f}}}-1)/(K+2)$, giving
$s^2=\log(1+K/3)$, fine-scale memory
$a_{\mathrm{f}}=\log\!\bigl(1+\rho_1(K+2)\bigr)/s^2$, and coarse memory
$a=a_{\mathrm{f}}^{\,d/d_c}$. We flag the substitution plainly: \emph{v2
as implemented runs on Clark's log-normal clock, not on the
tempered-stable chain of Definition~\ref{def:dspm}}---the choice buys the
closed-form Gaussian KL in the objective below, and it means the
tempered-stable machinery of
Sections~3.1--3.3 enters v2 only through the calibration logic and the
initialization; carrying the TS prior itself into v2 requires Monte-Carlo
KL and is future work. The constructive contribution (TS family,
calibration, limits) and the coupling contribution therefore currently
live on two adjacent members of the same variance-mixture class, not on
one model. The objective is
\begin{equation}\label{eq:v2loss}
\mathcal{L} \;=\;
\underbrace{\E\,\bigl\|z-\hat z_\phi\bigr\|^2}_{\text{denoising}}
\;+\;
\underbrace{\E_{q_\psi}\Bigl[\tfrac1{2d}\textstyle\sum_i
\bigl(x_{0,i}^2/A_i+\log A_i\bigr)\Bigr]
+\tfrac1d\,\mathrm{KL}\bigl(q_\psi(\mu^c\mid x_0)\,\big\|\,
p_{a,s}(\mu^c)\bigr)}_{\text{SV-ELBO}},
\end{equation}
the second group being the exact variational objective of the classical
stochastic-volatility state-space model $x_{0,i}=\sqrt{A_i}\,
\varepsilon_i$ with closed-form Gaussian KL.

\begin{remark}[Why the simplified loss cannot train the encoder]
\label{rem:logdet}
The $\varepsilon$-parametrized denoising term in \eqref{eq:v2loss} is the
Euclidean norm of the \emph{whitened} residual
(Proposition~\ref{prop:elbo}): whitening cancels the noise covariance, so
the term is nearly scale-free in $A$ and carries almost no gradient toward
matching $A$ to the path's realized variance. In the exact ELBO that
information lives in the log-determinant $\tfrac12\sum_i\log A_i$ of the
Gaussian likelihood---the term every simplified diffusion loss drops. We
verified both failure modes this analysis predicts: training
\eqref{eq:v2loss} without the SV term at KL weight $\beta=1$ collapses the
posterior ($\mathrm{KL}\to0.007$ nats per latent; interventions remain
ineffective), while $\beta=0.1$ makes the latent copy $x_0$
(reconstruction MSE drops to $0.02$; unconditional generation collapses to
excess kurtosis $0.18$). With the SV term, no $\beta$-tuning is needed.
\end{remark}

Generation is unchanged (sample the coarse prior chain, interpolate, run
the reverse diffusion given $A$); but because training pairs are now
coupled, $p_\theta(x_0\mid A)$ genuinely depends on $A$, and designed
variance paths become control inputs rather than nuisance.

\section{Experiments}\label{sec:exp}

Setup common to all runs: paths of length $d=64$; $T=100$, linear
$\beta\in[10^{-3},0.2]$; denoiser $=$ six dilated \texttt{Conv1d} residual
blocks with FiLM time conditioning ($\approx0.5$M parameters,
receptive field $69\ge d$); $6{,}000$ Adam steps, batch $256$;
Apple-Silicon laptop, minutes per model. Every headline comparison
(Tables~\ref{tab:results} and~\ref{tab:v2}) reports mean $\pm$ sd over
\emph{three training seeds} with the dataset held fixed; the sampler and
identity checks (Sections~\ref{sec:exp1}--\ref{sec:exp2}) have their own
replication protocols, and the blind and receptive-field ablations are
single-run (flagged where cited). Ground truth for
Sections~\ref{sec:exp3}--\ref{sec:exp5}: $28{,}000$ paths
($20{,}000$ train, $8{,}000$ test) of GARCH(1,1)-$t$
($\omega=0.05,a=0.10,b=0.85,\nu=6$), with pooled test statistics: excess
kurtosis $6.31$, Hill index $3.69$ (top $2.5\%$ of $|r|$),
$\rho_{r^2}(1)=0.159$, $\rho_{r^2}(5)=0.110$. Training-set estimates
$(\hat K,\hat\rho_1)=(6.27,0.161)$ calibrate the DSPM by
\eqref{eq:invert}: $\varphi=0.638$, $\theta=0.0876$; the i.i.d.-TS variant
is kurtosis-matched ($\varphi=0$, $\theta=0.144$), the shared-scale
variant is a kurtosis-matched inverse-gamma (equivalent Student-$t$ with
$\nu\approx4.96$).

\subsection{Sampler}\label{sec:exp1}

We sample $\mathrm{TS}(\alpha,\theta,\delta)$ by Kanter's one-sided stable
representation \citep{kanter1975,devroye1986} with exponential-tilting
rejection, splitting into $n=\lceil\delta\theta^\alpha\rceil$ pieces to
keep acceptance $\ge e^{-1}$ \citep{baeumer2010}. With $2\times10^6$
draws, the empirical Laplace transform matches \eqref{eq:lt} to
$\le2.1\times10^{-4}$ uniformly on $u\in[0.05,4]$ across untempered
($\alpha\in\{0.5,0.7\}$) and tempered cases including the splitting regime
$\delta\theta^\alpha\approx4.9$ ($n=5$); sample mean and variance match
$\kappa_1,\kappa_2$ to three digits (Figure~\ref{fig:sampler}).

\begin{figure}[t]
\centering\includegraphics[width=.92\textwidth]{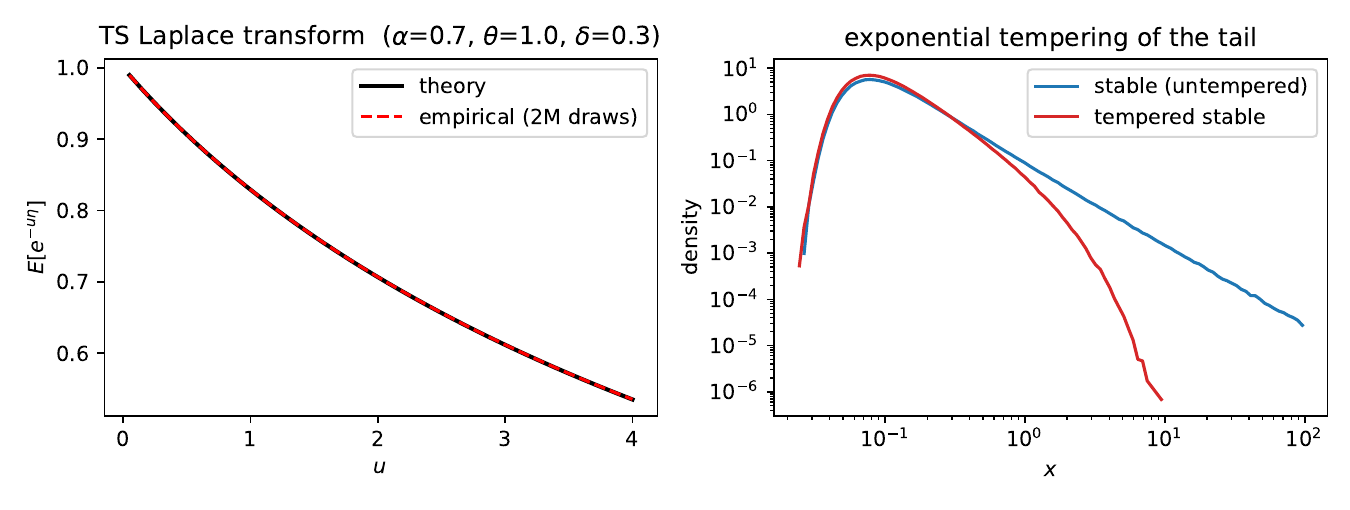}
\caption{Sampler validation. Left: empirical vs.\ exact Laplace transform,
$\mathrm{TS}(0.7,1.0,0.3)$. Right: tempering suppresses the far tail while
preserving the stable body.}
\label{fig:sampler}
\end{figure}

\subsection{Identities and calibration recovery}\label{sec:exp2}

Table~\ref{tab:identities} verifies Proposition~\ref{prop:moments} on a
$(\theta,\varphi)$ grid; Figure~\ref{fig:calib} shows the exact ACF curve,
the recovery experiment, and one noise path. At targets
$(K^\ast,\rho_1^\ast)=(8,0.20)$, \eqref{eq:invert} gives
$(\varphi^\ast,\theta^\ast)=(0.750,0.0643)$; re-estimating from finite
panels ($2{,}000$ paths, $d=64$) and re-inverting yields
$\hat\varphi=0.756\pm0.085$, $\hat\theta=0.066\pm0.007$ over $200$
replications: unbiased, exactly identified.

\begin{table}[t]
\centering
\caption{Proposition~\ref{prop:moments} vs.\ Monte Carlo ($8{,}000$ paths,
$d=64$, $\alpha=0.7$, pooled estimators).}
\label{tab:identities}
\begin{tabular}{cc cc cc}
\toprule
$\theta$ & $\varphi$ & $K$ th. & $K$ MC & $\rho_1$ th. & $\rho_1$ MC\\
\midrule
0.06 & 0.75 & 8.57 & 8.14 & 0.203 & 0.202\\
0.15 & 0.50 & 4.00 & 4.38 & 0.111 & 0.110\\
0.40 & 0.90 & 1.18 & 1.17 & 0.112 & 0.113\\
0.80 & 0.30 & 0.87 & 0.84 & 0.030 & 0.034\\
\bottomrule
\end{tabular}
\end{table}

\begin{remark}[Estimator bias under heavy tails]\label{rem:mikosch}
Averaging per-path short-window ACF estimates underestimates
$\rho_{r^2}$ by up to a factor $3$ at these parameters---the downward bias
of sample autocorrelations of heavy-tailed squares \citep{mikosch2000}.
All ACFs here use the pooled (global-mean, global-variance) estimator,
applied identically to data and models; with it, Monte Carlo matches
\eqref{eq:acf} to the third digit. Calibrations of \eqref{eq:invert} to
per-asset estimates should pool or bias-correct first.
\end{remark}

\begin{figure}[t]
\centering\includegraphics[width=\textwidth]{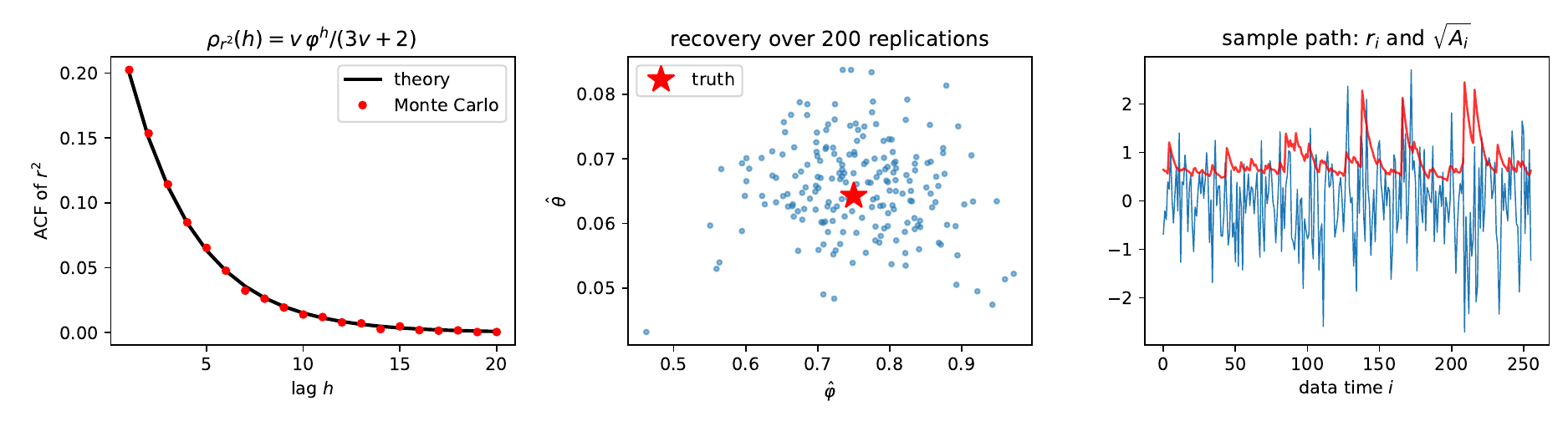}
\caption{Left: exact ACF \eqref{eq:acf} vs.\ Monte Carlo. Middle:
recovered $(\hat\varphi,\hat\theta)$ over $200$ replications (star $=$
truth). Right: noise path $r_i=\sqrt{A_i}z_i$ with latent
$\sqrt{A_i}$---clustering before any learning.}
\label{fig:calib}
\end{figure}

\subsection{Conditioned models: the invariance in action}\label{sec:exp3}

Four models differing \emph{only} in the law of $A$
(Table~\ref{tab:results}, Figure~\ref{fig:results}); all denoisers receive
$\sqrt A$ and whitened inputs; all metrics are pooled over $4{,}000$
generated paths, reported as mean $\pm$ sd over three training seeds.
Two findings. \textbf{(a) Dependence structure:}
all four seed-averaged models land within $0.017$ of the data's
$\rho_{r^2}(1)=0.159$ and within $0.011$ of its extremogram, regardless
of whether their noise offers no clustering ($A\equiv1$, i.i.d.),
constant clustering (shared-IG) or calibrated clustering (DSPM); the
between-model spread is comparable to the between-seed spread. With a
receptive-field-restricted denoiser (dilations $(1,1)$, RF $=13\ll64$;
single run) the picture is unchanged: small-RF DDPM reaches
$\rho_1=0.131$, small-RF DSPM $0.139$. (We note the restriction is
suggestive rather than a strict information bound: composing $T$ reverse
steps extends the effective receptive field well beyond a single pass, so
a small-RF sampler can still build long-range dependence iteratively.)
Clustering in the conditioned regime is \emph{learned}, not
transmitted---consistent with Proposition~\ref{prop:invariance}. The
decay shape confirms the source: conditioned-DSPM
$\rho_{r^2}(5)=0.118\pm0.012$ tracks the data's $0.110$ (learned), not
the mechanism's theoretical $0.027$ (see Section~\ref{sec:exp5}).
\textbf{(b) Marginals:} here seed variance forces us to retract the
sharper single-seed reading. Our first run showed DDPM clearly worst on
every tail metric; across three seeds the kurtosis differences dissolve
(DDPM $4.59\pm0.53$ against heavy-tailed variants' $4.3$--$5.2$, all
overlapping) and $W_1$ retains only a directional, within-noise ordering
(iid-TS $0.031\pm0.018$ and DSPM $0.038\pm0.013$ vs.\ DDPM
$0.055\pm0.028$). The one tail statistic that is simultaneously accurate
and stable is DSPM's Hill index, $3.70\pm0.04$ against the data's
$3.69$ (DDPM: $3.85\pm0.16$). At this scale, then, the inductive-bias
channel that Proposition~\ref{prop:invariance} leaves open is
\emph{at most marginal} in our measurements---the invariance is, if
anything, more complete than the theorem strictly requires.

\begin{table}[t]
\centering
\caption{Conditioned four-way comparison: mean $\pm$ sd over three
training seeds ($4{,}000$ generated paths per run; pooled estimators;
extremogram $=P(|r_{i+h}|>u\mid|r_i|>u)$ at the $95\%$ quantile $u$).}
\label{tab:results}
\begin{tabular}{l cccccc}
\toprule
 & ex.kurt & Hill & $\rho_{r^2}(1)$ & $W_1$ & xgram$_1$ & xgram$_5$\\
\midrule
data (test) & 6.31 & 3.69 & 0.159 & -- & 0.149 & 0.128\\
DDPM ($A\equiv1$) & $4.59{\pm}.53$ & $3.85{\pm}.16$ & $.142{\pm}.012$
 & $.055{\pm}.028$ & $.138{\pm}.010$ & $.122{\pm}.008$\\
iid-TS (DLPM-like) & $4.75{\pm}.11$ & $3.82{\pm}.03$ & $.149{\pm}.001$
 & $.031{\pm}.018$ & $.142{\pm}.002$ & $.123{\pm}.005$\\
shared-IG (t-EDM-like) & $4.33{\pm}.44$ & $3.88{\pm}.09$ & $.142{\pm}.004$
 & $.049{\pm}.023$ & $.142{\pm}.001$ & $.124{\pm}.006$\\
DSPM (TS-OU chain) & $5.23{\pm}.71$ & $3.70{\pm}.04$ & $.154{\pm}.011$
 & $.038{\pm}.013$ & $.146{\pm}.005$ & $.129{\pm}.005$\\
\bottomrule
\end{tabular}
\end{table}

\begin{figure}[t]
\centering\includegraphics[width=\textwidth]{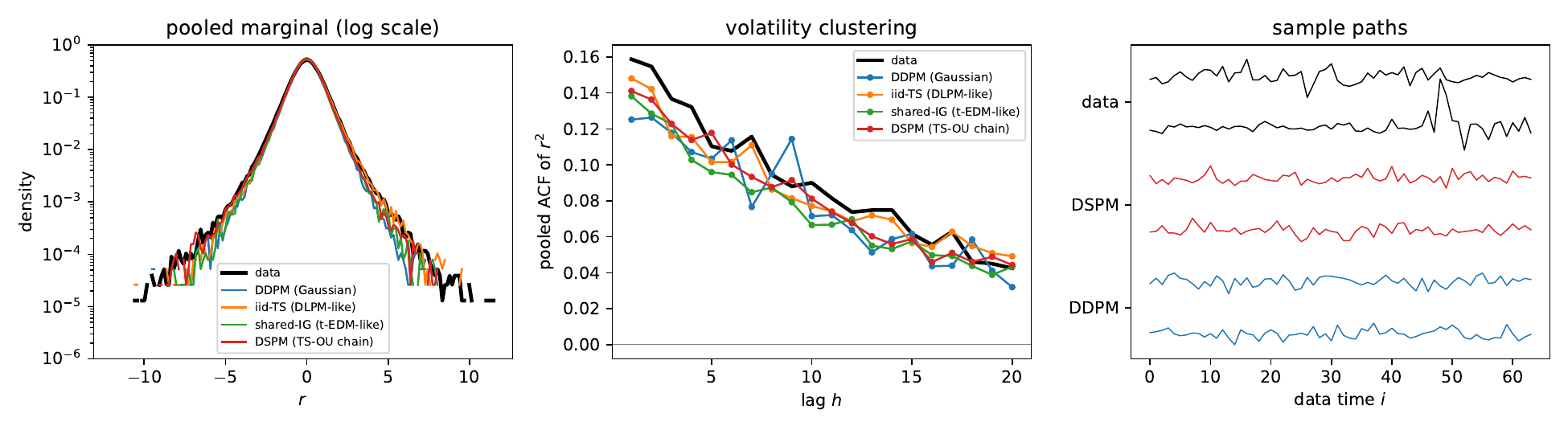}
\caption{Conditioned comparison. Left: pooled marginals (log scale).
Middle: pooled ACF of squared returns. Right: sample paths.}
\label{fig:results}
\end{figure}

\subsection{Interventions on $A$ do nothing}\label{sec:exp4}

Proposition~\ref{prop:invariance}(i) predicts that \emph{designing} the
variance chain at sampling time cannot steer a well-trained conditioned
model. We feed the trained DSPM denoiser three scenarios: the
unconditional chain; a volatility shock $A_{i}=1+7\varphi^{\,i-i_0}$ for
$i\ge i_0=24$ (peak $\times8$, the chain's own conditional-mean decay);
and a flat calm regime $A\equiv0.25$. Naive $\sqrt A$-scaling predicts
envelope changes of $+183\%$ at the peak and $-50\%$ in the calm regime.
Measured (Figure~\ref{fig:invariance}, one run shown; ratios averaged
over three training seeds): the envelope ratio at the shock peak is
$0.88\pm0.03$ against the naive $2.83$---a small dip with the wrong sign
for the naive prediction---and the calm-regime ratio is $1.06\pm0.01$
against the naive $0.50$. The trained network cancels $A$ almost exactly,
confirming that in v1 the latent is nuisance machinery, not a control
variable; scenario control requires coupling $A$ to the data
(Section~\ref{sec:v2}).

\begin{figure}[t]
\centering\includegraphics[width=.95\textwidth]{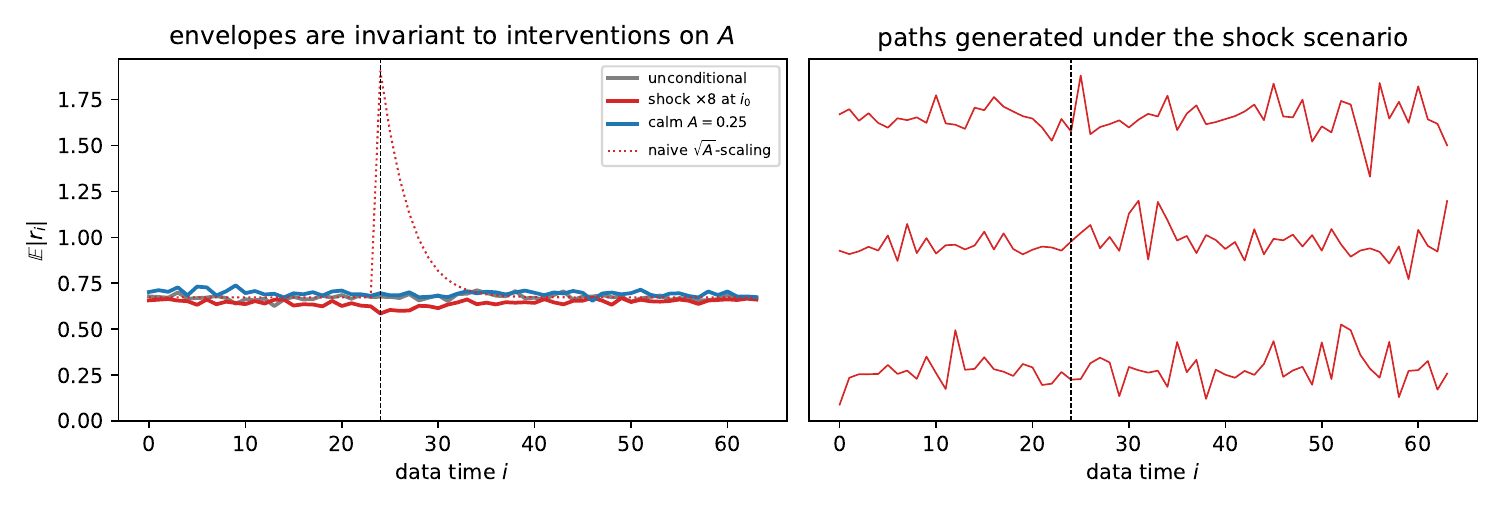}
\caption{Invariance under interventions. Left: amplitude envelopes under
the three designed $A$ scenarios coincide, far from the naive
$\sqrt A$-scaling (dotted). Right: paths generated under the $\times8$
shock scenario show no burst at $i_0$ (dashed line).}
\label{fig:invariance}
\end{figure}

\subsection{Blind models: the mechanism transmits, as calibrated}
\label{sec:exp5}

Finally we flip the single design bit the theorem identifies: the denoiser
is \emph{blind}---raw $x_t$ input, no $\sqrt A$ channel---while the
sampler machinery is unchanged. (This is an ablation isolating the
conditioning bit, not a faithful DLPM re-implementation: pairing a
marginal denoiser with $A$-scaled updates is deliberately crude, and
inflates marginals. Single training run; the effect sizes here are an
order of magnitude above the seed variability of
Table~\ref{tab:results}.) Table~\ref{tab:blind}: the blind DSPM's samples show
$\rho_{r^2}(1)=0.151$ against the \emph{calibrated} $0.161$, with
$\rho_{r^2}(5)=0.019$ against the mechanism's theoretical
$v\varphi^5/(3v+2)=0.027$---the geometric fingerprint of
\eqref{eq:acf}, clearly distinct from the data's $0.110$ and from the
conditioned model's learned $0.118$. The blind iid variant's clustering
collapses to $0.008$ and its extremogram to the $5\%$ base rate: an
i.i.d.\ mechanism has nothing to transmit. Meanwhile both blind models'
kurtosis explodes ($\approx70$): without whitening, the heavy injected
noise cannot be cleanly removed. Together with
Sections~\ref{sec:exp3}--\ref{sec:exp4} this completes the causal picture:
\emph{the law of $A$ determines what structure the noise offers
(Proposition~\ref{prop:moments}, quantitatively visible in transmission);
the conditioning choice determines whether it transmits
(Proposition~\ref{prop:invariance}).}

\begin{table}[t]
\centering
\caption{Conditioning ablation. ``Offered'' $=$ the noise mechanism's own
values from Proposition~\ref{prop:moments} at the calibrated parameters.}
\label{tab:blind}
\begin{tabular}{l ccccc}
\toprule
 & ex.kurt & Hill & $\rho_{r^2}(1)$ & $\rho_{r^2}(5)$ & xgram$_1$\\
\midrule
data (test)              & 6.31  & 3.69 & 0.159 & 0.110 & 0.149\\
offered by DSPM noise    & 6.27  & --   & 0.161 & 0.027 & --\\
\midrule
conditioned DSPM         & 4.26  & 3.76 & 0.141 & 0.118 & 0.140\\
blind DSPM               & 71.7  & 1.80 & 0.151 & 0.019 & 0.346\\
blind iid-TS             & 68.9  & 1.68 & 0.008 & 0.003 & 0.066\\
\bottomrule
\end{tabular}
\end{table}

\subsection{The coupled model: control restored, filter validated}
\label{sec:exp6}

We train v2 (Section~\ref{sec:v2}) with the same denoiser, schedule,
budget and data as v1, and re-run the intervention of
Section~\ref{sec:exp4} \emph{with the identical designed variance paths}
on both models (Figure~\ref{fig:v2}, Table~\ref{tab:v2}; all ratios and
metrics mean $\pm$ sd over three training seeds). Three results.

\textbf{(a) Control.} The $\times8$ shock, which moves the v1 envelope by
a factor $0.88\pm0.03$ (i.e.\ nothing, with the wrong sign), scales the
v2 envelope at the shock point by $\mathbf{3.00\pm0.14}$---bracketing the
naive $\sqrt A$ prediction of $2.83$---and the response decays along the
designed path ($1.98\to1.51\to1.07\to0.90$ over four steps in the run
shown). The calm scenario $A\equiv0.25$ scales the envelope to
$0.52\pm0.01$ times baseline (naive: $0.50$; v1: $1.06\pm0.01$). The same
experiment, the same intervention, opposite outcome, stable across
seeds: the difference is one premise of
Proposition~\ref{prop:invariance}.

\textbf{(b) Filtering.} On held-out paths the encoder's posterior mean of
$A$ attains a mean \emph{per-path} correlation of $0.567\pm0.001$ with
the true GARCH variance---known exactly because the data is
synthetic---and $0.762\pm0.001$
pooled. The per-path figure is the honest measure of within-path
tracking; the pooled figure is inflated by cross-path differences in
overall volatility level, which are easier to infer than dynamics. A
$0.57$ within-path correlation makes the encoder a serviceable but
unremarkable volatility smoother; its value here is that it is validated
against ground truth rather than a proxy, and that it exists at all where
the simplified objective produced none.

\textbf{(c) Memory, learned rather than calibrated.} Empirical Bayes moves
the prior's memory from its moment-calibrated initialization ($0.638$ at
fine scale) to $0.803\pm0.000$ (identical to three decimals across
seeds), in the direction of the ground-truth persistence
$a+b=0.95$. We state the comparison's limits: the three numbers live in
three related but distinct parametrizations---the TS chain's $\varphi$,
the coarse log-normal $a$ re-expressed at fine scale through the
interpolation, and GARCH's $a+b$---so the claim is directional, not
metric. The direction itself is the classical errors-in-variables
prediction: calibrating persistence to the autocorrelation of
\emph{squared returns}, a noisy proxy of latent variance, attenuates it,
and likelihood-based inference partially undoes the attenuation.

The measured cost closes the ledger, and it is not small: v2's
unconditional excess kurtosis of $2.86\pm0.12$ falls short not only of
the target $6.31$ and of v1's $5.23\pm0.71$, but of the \emph{Gaussian}
baseline's $4.59\pm0.53$---a gap far outside seed noise; its
$\rho_{r^2}(1)=0.114\pm0.003$ is likewise below all four conditioned
models ($0.142$--$0.154$). The paper's opening argument was that Gaussian noise discards the
domain's structure---and the one model here that achieves control
currently reproduces that structure worst, because the coarse bottleneck
and the learned log-normal prior smooth away variance extremes. The
honest summary is that v1 delivers fidelity without control and v2
delivers control without fidelity; obtaining both in one model is the
open engineering problem (finer latent with a residual fine-scale
variance component; tempered-stable prior via Monte-Carlo KL), not an
achievement of this paper.

\begin{table}[t]
\centering
\caption{v1 vs.\ v2 under identical interventions and identical
unconditional metrics. ``Shock ratio'' $=$ envelope at the shock peak over
baseline (naive $\sqrt A$ scaling: $2.83$); ``calm ratio'' analogous
(naive: $0.50$). The three memory-column entries are expressed in
related but distinct parametrizations (TS-chain $\varphi$; coarse
log-normal $a$ at fine scale; GARCH $a+b$)---see the caveat in
Section~\ref{sec:exp6}(c).}
\label{tab:v2}
\begin{tabular}{l cc c c ccc}
\toprule
 & shock & calm & filter & fine-scale & \multicolumn{3}{c}{unconditional}\\
 & ratio & ratio & corr & memory $a$ & ex.kurt & $\rho_{r^2}(1)$ & $W_1$\\
\midrule
target / naive & 2.83 & 0.50 & 1 & 0.95 & 6.31 & 0.159 & --\\
v1 (uncoupled) & $.88{\pm}.03$ & $1.06{\pm}.01$ & -- & 0.638 (calib.)
 & $5.23{\pm}.71$ & $.154{\pm}.011$ & $.038{\pm}.013$\\
v2 (coupled) & $3.00{\pm}.14$ & $.52{\pm}.01$ & $.57{\pm}.00$
 & 0.803 (learned) & $2.86{\pm}.12$ & $.114{\pm}.003$ & $.071{\pm}.006$\\
\bottomrule
\end{tabular}
\end{table}

\begin{figure}[t]
\centering\includegraphics[width=\textwidth]{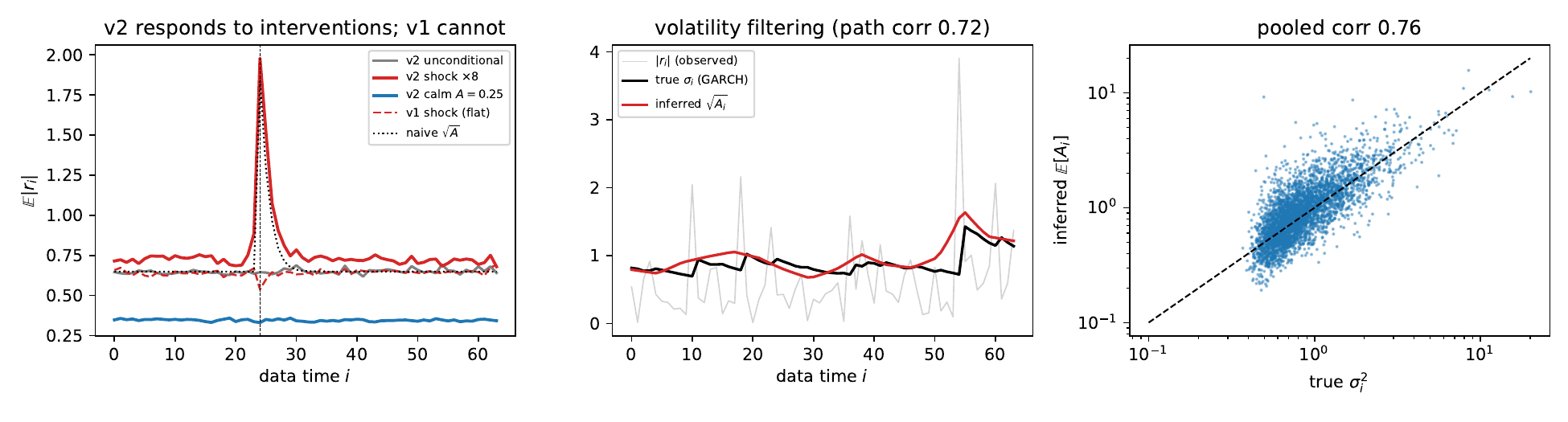}
\caption{v2. Left: under identical designed interventions, the v2 envelope
responds along the designed decay while v1 stays flat. Middle: inferred
$\sqrt{A_i}$ vs.\ the true GARCH $\sigma_i$ on a held-out path. Right:
pooled inferred variance vs.\ truth.}
\label{fig:v2}
\end{figure}

\section{Physics correspondence}\label{sec:physics}

The DSPM noise prior is a known object in statistical physics, and the
correspondences are exact, not analogies. \textbf{Superstatistics}
\citep{beck2003}: $r=\sqrt AZ$ with slowly varying $A$ is a
superstatistical ensemble; $\chi^2$-distributed intensity gives
Student-$t$ (t-EDM), our chain gives a self-decomposable tempered mixing
law. \textbf{Diffusing diffusivity}: equations
\eqref{eq:chain}--\eqref{eq:forward} are a discrete diffusing-diffusivity
model \citep{chubynsky2014}---a Brownian particle whose diffusion
coefficient is itself a stationary process---the accepted mechanism for
``Brownian yet non-Gaussian'' transport; the subordination analysis of
\citet{chechkin2017} transfers verbatim and predicts re-Gaussianization of
aggregated noise over horizons $\gg(1-\varphi)^{-1}$, i.e.\ stylized fact
(iv). \textbf{Subordination}: the chain is a random clock rate in the
sense of \citet{clark1973} and of continuous-time random walks
\citep{montroll1965}; DSPM is a neural generative model driven by a
discretized time-changed Brownian motion, the minimal adequate description
of returns according to \citet{carr2003sv}. These statements constrain the
\emph{noise mechanism} only; by Proposition~\ref{prop:invariance} they
imply nothing about the learned model in the conditioned regime---a
caution the physics framing makes easy to forget.

\section{Related work}\label{sec:related}

Heavy-tailed diffusion: DLPM \citep{shariatian2024dlpm}, LIM
\citep{yoon2023lim} (i.i.d.\ stable mixing; no memory); t-EDM
\citep{pandey2025tedm} (shared scale; infinite memory); fractional
generative diffusion \citep{nobis2023} (Gaussian long memory);
jump-diffusion forward processes \citep{baule2026}; GBM-based financial
diffusion \citep{kim2025gbm}, where heteroscedasticity is a function of
the observed price level rather than a latent process with calibratable
memory. None places a dynamically correlated latent variance inside the
noise, and none states the invariance of Section~\ref{sec:invariance}. We
are careful about that result's blast radius: it constrains mechanism
claims only in the \emph{conditioned} regime (denoiser shown the mixing
variables, drawn independently of the data). Models that operate blind
with respect to their mixing variables---as published DLPM effectively
does---fall on the transmission side of the dichotomy, where the
mechanism does act; for them the relevant finding here is different,
namely that an i.i.d.\ mixing structure has no temporal dependence to
transmit.
On the classical side DSPM is the neural counterpart of BNS
\citep{bns2001} and time-changed L\'evy models \citep{carr2003sv};
simulation follows \citet{kanter1975,baeumer2010,sabino2022}.

\section{Limitations and outlook}\label{sec:limits}

\textbf{(1) Single-factor memory.} One OU factor forces geometric ACF
decay; the calibrated $\varphi=0.638$ cannot match the GARCH truth's
persistence $a+b=0.95$ at long lags (visible in
Table~\ref{tab:blind}: offered $\rho_{r^2}(5)=0.027$ vs.\ data $0.110$).
Superpositions of OU factors (supOU) fix this within the same
propositions; consistently, v2's \emph{learned} memory ($0.803$) moved
toward but not onto the truth. \textbf{(2) The v2 trade-off is not yet
optimized.} The coarse bottleneck ($d_c=16$) that blocks the copy-cheat
also smooths away variance extremes, costing unconditional kurtosis
($2.69$ vs.\ v1's $4.26$); and the log-normal prior was chosen for its
closed-form KL, not its tails. A finer latent with a residual fine-scale
variance component, and the tempered-stable prior of
Definition~\ref{def:dspm} with Monte-Carlo KL, are the natural v2.1.
\textbf{(3) The blind sampler is crude.} The blind rows of
Table~\ref{tab:blind} isolate the conditioning bit; they are not a fair
evaluation of DLPM, whose reverse process is derived for the marginal
noise. \textbf{(4) Synthetic ground truth.} GARCH-$t$ was chosen so the
target stylized facts are known and controllable; the economically
decisive test---pricing path-dependent structures under generated
measures, where post-breach volatility persistence carries P\&L---is left
for the next iteration, on real index data. \textbf{(5) Scale.} $d=64$,
$0.5$M parameters, minutes of training: deliberately small, fully
reproducible; conclusions about large-scale behaviour require care,
though the invariance theorem itself is scale-free.

\bigskip
\noindent\textbf{Reproducibility.} \texttt{dspm\_core.py} (sampler, chain,
identities, calibration), \texttt{dspm\_net.py} (denoiser, sampler),
scripts \texttt{01}--\texttt{09} (validation, calibration, v1 training,
evaluation, invariance, blind ablation, v2 training and evaluation); fixed
seeds throughout.


\begin{thebibliography}{99}\small

\bibitem[Baeumer and Meerschaert(2010)]{baeumer2010}
B.~Baeumer and M.~M.~Meerschaert.
Tempered stable L\'evy motion and transient super-diffusion.
\emph{J. Comput. Appl. Math.}, 233:2438--2448, 2010.

\bibitem[Barndorff-Nielsen and Shephard(2001)]{bns2001}
O.~E.~Barndorff-Nielsen and N.~Shephard.
Non-Gaussian Ornstein--Uhlenbeck-based models and some of their uses in
financial economics.
\emph{J. R. Stat. Soc. B}, 63(2):167--241, 2001.

\bibitem[Baule(2026)]{baule2026}
R.~Baule.
Score-based generative models with jump-diffusion forward processes.
Working paper, 2026.

\bibitem[Beck and Cohen(2003)]{beck2003}
C.~Beck and E.~G.~D.~Cohen.
Superstatistics.
\emph{Physica A}, 322:267--275, 2003.

\bibitem[Bollerslev(1986)]{bollerslev1986}
T.~Bollerslev.
Generalized autoregressive conditional heteroskedasticity.
\emph{J. Econometrics}, 31:307--327, 1986.

\bibitem[Carr et~al.(2002)]{cgmy2002}
P.~Carr, H.~Geman, D.~B.~Madan, and M.~Yor.
The fine structure of asset returns: an empirical investigation.
\emph{J. Business}, 75(2):305--332, 2002.

\bibitem[Carr et~al.(2003)]{carr2003sv}
P.~Carr, H.~Geman, D.~B.~Madan, and M.~Yor.
Stochastic volatility for L\'evy processes.
\emph{Math. Finance}, 13(3):345--382, 2003.

\bibitem[Chechkin et~al.(2017)]{chechkin2017}
A.~V.~Chechkin, F.~Seno, R.~Metzler, and I.~M.~Sokolov.
Brownian yet non-Gaussian diffusion: from superstatistics to subordination
of diffusing diffusivities.
\emph{Phys. Rev. X}, 7:021002, 2017.

\bibitem[Chubynsky and Slater(2014)]{chubynsky2014}
M.~V.~Chubynsky and G.~W.~Slater.
Diffusing diffusivity: a model for anomalous, yet Brownian, diffusion.
\emph{Phys. Rev. Lett.}, 113:098302, 2014.

\bibitem[Clark(1973)]{clark1973}
P.~K.~Clark.
A subordinated stochastic process model with finite variance for
speculative prices.
\emph{Econometrica}, 41(1):135--155, 1973.

\bibitem[Cont(2001)]{cont2001}
R.~Cont.
Empirical properties of asset returns: stylized facts and statistical
issues.
\emph{Quant. Finance}, 1:223--236, 2001.

\bibitem[Devroye(1986)]{devroye1986}
L.~Devroye.
\emph{Non-Uniform Random Variate Generation}.
Springer, 1986.

\bibitem[Engle(1982)]{engle1982}
R.~F.~Engle.
Autoregressive conditional heteroscedasticity with estimates of the
variance of United Kingdom inflation.
\emph{Econometrica}, 50(4):987--1007, 1982.

\bibitem[Gopikrishnan et~al.(1999)]{gopikrishnan1999}
P.~Gopikrishnan, V.~Plerou, L.~A.~N.~Amaral, M.~Meyer, and H.~E.~Stanley.
Scaling of the distribution of fluctuations of financial market indices.
\emph{Phys. Rev. E}, 60:5305--5316, 1999.

\bibitem[Ho et~al.(2020)]{ho2020ddpm}
J.~Ho, A.~Jain, and P.~Abbeel.
Denoising diffusion probabilistic models.
\emph{NeurIPS}, 33:6840--6851, 2020.

\bibitem[Kanter(1975)]{kanter1975}
M.~Kanter.
Stable densities under change of scale and total variation inequalities.
\emph{Ann. Probab.}, 3(4):697--707, 1975.

\bibitem[Kim et~al.(2025)]{kim2025gbm}
Kim et~al.
Diffusion-based generative models for financial time series via geometric
Brownian motion.
arXiv:2507.19003, 2025.

\bibitem[Koponen(1995)]{koponen1995}
I.~Koponen.
Analytic approach to the problem of convergence of truncated L\'evy
flights towards the Gaussian stochastic process.
\emph{Phys. Rev. E}, 52:1197--1199, 1995.

\bibitem[Mantegna and Stanley(1994)]{mantegna1994}
R.~N.~Mantegna and H.~E.~Stanley.
Stochastic process with ultraslow convergence to a Gaussian: the truncated
L\'evy flight.
\emph{Phys. Rev. Lett.}, 73:2946--2949, 1994.

\bibitem[Mikosch and St\u{a}ric\u{a}(2000)]{mikosch2000}
T.~Mikosch and C.~St\u{a}ric\u{a}.
Limit theory for the sample autocorrelations and extremes of a GARCH(1,1)
process.
\emph{Ann. Statist.}, 28(5):1427--1451, 2000.

\bibitem[Montroll and Weiss(1965)]{montroll1965}
E.~W.~Montroll and G.~H.~Weiss.
Random walks on lattices. II.
\emph{J. Math. Phys.}, 6:167--181, 1965.

\bibitem[Nobis et~al.(2023)]{nobis2023}
G.~Nobis et~al.
Generative fractional diffusion models.
arXiv:2310.17638, 2023.

\bibitem[Pandey et~al.(2025)]{pandey2025tedm}
K.~Pandey et~al.
Heavy-tailed diffusion models.
\emph{ICLR}, 2025. arXiv:2410.14171.

\bibitem[Sabino and Cufaro~Petroni(2022)]{sabino2022}
P.~Sabino and N.~Cufaro~Petroni.
Fast simulation of tempered stable Ornstein--Uhlenbeck processes.
\emph{Comput. Statist.}, 37:2517--2551, 2022.

\bibitem[Shariatian et~al.(2024)]{shariatian2024dlpm}
D.~Shariatian, U.~Simsekli, and A.~Durmus.
Denoising L\'evy probabilistic models.
arXiv:2407.18609, 2024.

\bibitem[Sohl-Dickstein et~al.(2015)]{sohl2015}
J.~Sohl-Dickstein, E.~Weiss, N.~Maheswaranathan, and S.~Ganguli.
Deep unsupervised learning using nonequilibrium thermodynamics.
\emph{ICML}, 2015.

\bibitem[Song et~al.(2021)]{song2021sde}
Y.~Song, J.~Sohl-Dickstein, D.~P.~Kingma, A.~Kumar, S.~Ermon, and
B.~Poole.
Score-based generative modeling through stochastic differential equations.
\emph{ICLR}, 2021.

\bibitem[Yoon et~al.(2023)]{yoon2023lim}
E.~B.~Yoon, K.~Park, S.~Kim, and S.~Lim.
Score-based generative models with L\'evy processes.
\emph{NeurIPS}, 36, 2023.

\bibitem[Anonymous(2026)]{heavytails2026}
Do heavy tails help diffusion? A statistical-estimation perspective.
Working paper, 2026.

\end{thebibliography}
\end{document}